\documentclass[sigconf,nonacm]{acmart}
\setkeys{acmart.cls}{balance=false}
\usepackage{bm}
\usepackage{dsfont}  
\usepackage{multirow}
\usepackage{algorithm}
\usepackage{algpseudocode}
\usepackage{placeins}

\usepackage[capitalize,noabbrev]{cleveref}

\theoremstyle{acmplain}
\newtheorem{theorem}{Theorem}
\newtheorem{proposition}{Proposition}
\newtheorem{lemma}{Lemma}
\theoremstyle{acmdefinition}

\newcommand{\E}{\mathbb{E}}

\crefname{problem}{Problem}{Problems}
\crefname{theorem}{Theorem}{Theorems}
\crefname{proposition}{Proposition}{Propositions}
\crefname{lemma}{Lemma}{Lemmas}
\crefname{algorithm}{Algorithm}{Algorithms}

\ifdefined\pdfsuppressptexinfo\pdfsuppressptexinfo=-1\relax\fi

\AtBeginDocument{%
  }

\setcopyright{none}
\renewcommand{\footnotetextcopyrightpermission}[1]{}
\acmDOI{}
\acmISBN{}

\begin{document}
\title[Decoupling Procedural and Outcome Fairness]{Decoupling Procedural and Outcome Fairness in Seed Allocation for Influence Maximization}
\author{Xiaojing Du}
\affiliation{%
  \institution{Adelaide University}
  \city{Adelaide}
  \country{Australia}
}

\author{Chunlei Meng}
\affiliation{%
  \institution{Fudan University}
  \city{Shanghai}
  \country{China}
}

\author{Wentao Gao}
\affiliation{%
  \institution{Adelaide University}
  \city{Adelaide}
  \country{Australia}
}

\author{Ke Sun}
\affiliation{%
  \institution{Dalian University}
  \city{Dalian}
  \country{China}
}

\author{Feiyu Yang}
\affiliation{%
  \institution{Civil Aviation Flight University of China}
  \city{Guanghan}
  \country{China}
}

\renewcommand{\shortauthors}{Du et al.}

\begin{abstract}
On online social platforms, equal seed access to interventions can yield unequal group outcomes through social spillovers, so an access audit alone can miss outcome disparities. We characterise this trade-off through an exact allocation--structure decomposition and a sharp first-moment floor for graph-independent seeding under additive spillovers on stochastic block networks. Decoupled Procedural Allocation (DPA) combines group quotas with selection by seeds' contributions to each group's outcome. Rankings balance total outcome against disparity, and within-quota exchanges reduce the gap. Computable bounds quantify retained total outcome. DPA achieves the smallest mean outcome gap among the compared methods at the common access tolerances on both development and prospective suites. Exposure and independent-cascade studies show how the propagation mechanism changes this trade-off. On Pokec, DPA nearly equalises modelled gender outcomes at the closest access-parity quota while preserving most of the degree-based allocation's total outcome.
\end{abstract}

\maketitle

\begin{figure}[t]
\centering
\includegraphics[width=\columnwidth]{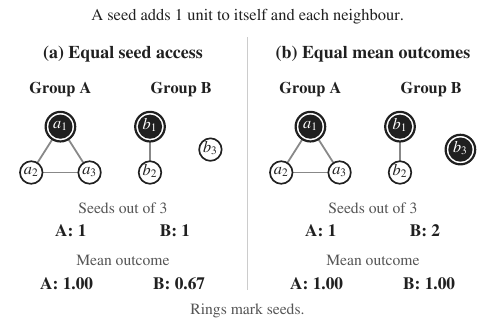}
\Description{Group A contains a triangle and group B contains an edge and an isolated node. With one seed in each group, their mean outcomes are one and two thirds. Adding B's isolated node as a seed equalises the means and increases B's seed share.}
\caption{Equal seed shares leave unequal group means. From zero baseline, each seed adds one unit to itself and each neighbour. Adding the isolated $B$ node equalises outcomes and raises its group seed rate from one third to two thirds. Rings mark seeds.}
\label{fig:toy}
\end{figure}

\section{Introduction}
\label{sec:intro}

Online social platforms grant selected users direct access to campaigns, promotions, or new services. Benefits propagate from these seed users through social ties to users without direct access. Influence maximization (IM) formalises the aggregate reach objective \citep{kempe2003}. Fair influence maximization (FIM) addresses benefit disparities that can leave groups underserved \citep{fish2019gaps,tsang2019,rahmattalabi2021welfare}.

Direct opportunities and propagated outcomes raise distinct fairness requirements. Group attributes are associated with different patterns of online social connectivity \citep{traud2012social}, so equal per-capita access can leave one group with fewer opportunities to receive spillovers. Conversely, equalising group outcomes can require unequal access to the initial intervention. We quantify attainable outcome disparity under access constraints and the total-outcome cost of reducing it. Since these spillovers travel through online friendship graphs whose structure shapes fairness \citep{saxena2024fairsna}, the problem belongs to the \emph{Social Networks and Social Media} track.

Existing FIM formulations include worst-group coverage and social welfare \citep{tsang2019,rahmattalabi2021welfare}, as well as proportional seed constraints \citep{farnadi2020unifying}. The split between seed constraints and outcome objectives mirrors procedural and distributive fairness \citep{greenberg1987}. Since reach depends on initial recipients' network positions \citep{banerjee2013diffusion,kim2015lancet,banerjee2019restud}, planners need to know the outcome disparity left by an access tolerance and the cost of reducing it.

\Cref{fig:toy} illustrates the structural difficulty. Group $A$ forms a triangle, while $B$ contains an edge and an isolated node. With one seed in each group, their mean outcomes are $1$ and $2/3$. Adding the isolated node equalises outcomes by giving $B$ more direct access. No nonempty seed set on this graph achieves both allocation parity and outcome parity.

Under additive network spillovers on a stochastic block model (SBM), the expected signed outcome difference of graph-independent seeding decomposes exactly into an allocation component and a structural component. This decomposition yields a sharp continuous lower bound on the magnitude of the expected difference under an access tolerance and identifies the allocation direction that offsets structural asymmetry. Comparing exposure models also distinguishes two mechanisms: count exposure accumulates the contributions of seeded neighbours, whereas fraction exposure averages them and substantially reduces the structural disparity.

Decoupled Procedural Allocation (DPA) uses this characterisation to separate quota choice from seed selection. On observed graphs, DPA evaluates seed contributions to both group means. Tilted rankings balance disparity against total outcome, and within-group exchanges refine the selection. An exact envelope excludes quotas that cannot meet the outcome tolerance, while upper bounds quantify retained total outcome. A closed-form quota rule provides a model-based reference with a randomised fairness guarantee.

\par\noindent
\begin{minipage}{\columnwidth}
We make three contributions.
\begin{itemize}
\item For graph-independent seeding, we derive an exact decomposition into allocation and structural terms, with a sharp continuous first-moment floor under an access tolerance and related identities for multiple groups and fraction exposure.
\item We develop closed-form quotas and observed-graph seed selection within quotas, with exact gap envelopes excluding infeasible quotas and total-outcome certificates for feasible solutions.
\item Experiments on synthetic and Web graphs, including a prospective suite fixed before evaluation and million-user Pokec, compare recent seed-level methods and exact integer-programming references to assess quotas, seed composition, and exposure.
\end{itemize}
\end{minipage}
\par

\section{Related Work}
\label{sec:related}

\paragraph{Fair influence maximization.}
Maximising aggregate reach \citep{kempe2003} can underserve structurally disadvantaged groups \citep{fish2019gaps,stoica2020biased}. Network topology also shapes disparities in information access \citep{jalali2023fairness}. FIM addresses outcome equity through maximin objectives \citep{tsang2019}, group welfare \citep{rahmattalabi2021welfare}, time-critical coverage \citep{ali2023fairness}, randomised policies \citep{becker2021randomization}, and scalable optimisation \citep{feng2023scale,rui2025tkde,chen2025sfimco}. Recent methods optimise multiple group objectives through linear programming (LP) \citep{spaeh2025multiobjective} or aggregate group-specific seed lists \citep{rui2025maximin}. The equality formulation of \citet{farnadi2020unifying} matches seed shares to population shares. \citet{becker2023cost} studies the cost of equal group coverage. \citet{stoica2020biased} analyse diversity-aware degree thresholds on biased preferential-attachment networks using a one-hop diffusion approximation. We derive an access-constrained first-moment floor for graph-independent seeding under additive spillovers on balanced stochastic block networks and compare quota-level and seed-level objectives under one outcome model.

\paragraph{Network interference.}
Causal inference under interference distinguishes direct and spillover effects \citep{hudgens2008toward}. Exposure mappings relate outcomes to neighbours' treatments \citep{aronow2017aoas,forastiere2021identification}. Graph cluster randomisation designs interventions around network exposure \citep{ugander2013graph}. Linear peer-effects models provide a complementary parameterisation \citep{manski1993reflection,bramoulle2009identification,bramoulle2020survey}, and recent work examines interference misspecification \citep{weinstein2026misspecified}. Our linear exposure model gives an explicit fairness decomposition from planner-supplied direct and spillover coefficients. Count and fraction exposure yield different structural coefficients, making propagation part of allocation.

\paragraph{Procedural and distributive fairness.}
Procedural justice concerns decision processes, while distributive justice concerns their benefits \citep{greenberg1987}. Algorithmic treatments include feature-selection criteria \citep{grgic2018beyond} and procedural-unfairness measures \citep{wang2024procedural}. In network seed allocation, equal per-capita treatment is one operational measure of access. We characterise how this access criterion relates to group outcomes under network spillovers.

\section{Problem Formulation}
\label{sec:formulation}

\paragraph{Population and allocation.}
Let $\mathcal G=(V,E)$ be a simple undirected graph on $n$ vertices, partitioned into nonempty groups $V_1,\ldots,V_G$ of sizes $n_g$. A seed set $S$ has fixed size $0\leq k\leq n$ and induces treatment indicators $T_i=\mathds{1}[i\in S]$. Write $N_i$ for the neighbours of $i$, $d_i=|N_i|$, and $d_g(i)=|N_i\cap V_g|$. The group allocation rates satisfy
\begin{equation}
\pi_g(S)=\frac{|S\cap V_g|}{n_g},\qquad \sum_g n_g\pi_g(S)=k.
\label{eq:pi-g}
\end{equation}
We develop the two-group case, $V_A\sqcup V_B=V$, and extend the decomposition to $G$ groups.

\paragraph{Outcomes under interference.}
We use the additive neighbourhood-exposure model
\begin{equation}
Y_i(S)=\alpha+\beta T_i+\gamma\sum_{j\in N_i}T_j+\xi_i,
\label{eq:outcome-model}
\end{equation}
where $\xi_i$ is outcome noise with $\E[\xi_i\mid\mathcal G,S]=0$, and $\alpha$ is a common baseline. Becoming a seed adds the direct gain $\beta>0$ to a user's score, and each seeded neighbour adds the spillover gain $\gamma\geq0$. The planner specifies these gains. Outcomes are scores and can exceed one. We study fraction exposure and independent-cascade (IC) activation separately to assess how the representation of spillover affects the allocation trade-off.

The group mean outcome and total outcome are
\begin{align}
\mu_g(S)&=\frac{1}{n_g}\sum_{i\in V_g}\E[Y_i(S)\mid\mathcal G,S],\label{eq:mu-g}\\
\sigma(S)&=n\alpha+\beta k+\gamma\sum_{i\in S}d_i.\label{eq:spread}
\end{align}
For a fixed graph and seed set, both are available in closed form. In the SBM analysis, $\E$ additionally averages over the graph and any independent seed randomisation.

\paragraph{Two fairness gaps.}
We measure differences in access and outcomes by
\begin{align}
\Delta_T(S)&=|\pi_A(S)-\pi_B(S)|,\label{eq:delta-T}\\
\Delta_Y(S)&=|\mu_A(S)-\mu_B(S)|.\label{eq:delta-Y}
\end{align}
For $G>2$, each gap is the maximum over group pairs. The access tolerance $\tau\geq0$ bounds the difference in seed rates. For example, $\tau=0.1$ permits at most a ten-percentage-point difference between the two rates. A planner with this tolerance solves
\begin{equation}
\min_{S\subseteq V}\Delta_Y(S)
\quad\text{subject to}\quad |S|=k,\quad\Delta_T(S)\leq\tau.
\label{eq:problem}
\end{equation}
To select efficiently among allocations with approximately equal group outcomes, we also study
\begin{equation}
\begin{aligned}
&\max_{S\subseteq V}\sigma(S)\\
&\text{subject to}\quad |S|=k,\quad\Delta_T(S)\leq\tau,\quad\Delta_Y(S)\leq\epsilon.
\end{aligned}
\label{eq:problem-band}
\end{equation}
The outcome tolerance $\epsilon\geq0$ sets the planner's parity resolution. We use $\epsilon_0=\beta/(4\max_g n_g)$ and evaluate its sensitivity experimentally. If the band is infeasible, Eq.~\eqref{eq:problem} defines the minimum-gap objective, with total outcome breaking ties. If no integer quota meets the access tolerance, we report the minimum attainable access gap and a quota attaining it. The corresponding continuous rate set is
\begin{equation}
\Pi(k,\tau)=\{\bm\pi\in[0,1]^2:n_A\pi_A+n_B\pi_B=k,\ |\pi_A-\pi_B|\leq\tau\}.
\label{eq:feasible-rates}
\end{equation} Population analysis assumes graph-independent allocation, while observed-graph selection uses exact group contributions.

\section{Structural Characterisation}
\label{sec:theory}

Consider a two-block SBM with within-group edge probabilities $p_A,p_B$ and cross-group probability $q$. For balanced groups, write $n_A=n_B=m=n/2$ and $\mathcal H=p_A-p_B\geq0$. Seed identities and their fixed group counts are independent of the sampled graph. This includes a fixed seed set or uniform sampling within each group.

\begin{proposition}[No spillover]
\label{prop:sutva}
When $\gamma=0$, every seed set satisfies $\Delta_Y(S)=\beta\Delta_T(S)$.
\end{proposition}

\begin{theorem}[Allocation--structure decomposition]
\label{thm:upper}
For graph-independent seeding with fixed group rates,
\begin{equation}
\E[\mu_A-\mu_B]=\kappa_1(\pi_A-\pi_B)+\kappa_2\pi_B,
\label{eq:decomp-thm}
\end{equation}
where $\kappa_1=\beta+\gamma[(n_A-1)p_A-n_Aq]$ and
\begin{equation}
\kappa_2=\gamma[(n_A-1)p_A-(n_B-1)p_B+(n_B-n_A)q].
\label{eq:structural-coefficient}
\end{equation}
For balanced groups, $\kappa_2=\gamma(m-1)\mathcal H$. If $D=\mu_A-\mu_B$ and $r_n=\E|D-\E D|$, then
\begin{equation}
|\E D|\leq\E\Delta_Y\leq|\E D|+r_n,
\qquad r_n=O(\gamma\sqrt{p_{\max}}),
\label{eq:fluctuation}
\end{equation}
in the balanced case, where $p_{\max}=\max\{p_A,p_B,q\}$.
\end{theorem}

The allocation component in Eq.~\eqref{eq:decomp-thm} changes with the rate difference. The structural component persists at equal positive rates whenever $\kappa_2\ne0$. Equation~\eqref{eq:fluctuation} connects this signed first moment to the absolute gap used for evaluation. At fixed densities, a positive seed fraction, and nonzero $\gamma\mathcal H$, the parity term grows with $n$, whereas the fluctuation bound stays bounded. Under $p_{\max}=O(\log n/n)$, the fluctuation bound itself vanishes. For the rate profile $\bm\pi=(\pi_A,\pi_B)$, define $\Phi(\bm\pi)=|\kappa_1(\pi_A-\pi_B)+\kappa_2\pi_B|$ as the continuous extension of the first-moment gap.

\begin{theorem}[Sharp first-moment floor]
\label{thm:lower}
For balanced groups, define $\bar\pi=k/n$, $c=\kappa_1-\kappa_2/2$, and $\tau_{\mathrm{eff}}=\min\{\tau,2\bar\pi,2(1-\bar\pi)\}$. The continuous minimum of the first-moment gap is
\begin{equation}
\mathcal L(\tau)=\min_{\bm\pi\in\Pi(k,\tau)}\Phi(\bm\pi)
=\bigl[\kappa_2\bar\pi-|c|\tau_{\mathrm{eff}}\bigr]_+,
\label{eq:floor}
\end{equation}
where $[x]_+=\max\{x,0\}$ and $\Pi(k,\tau)$ satisfies budget, capacity, and access constraints. Every graph-independent allocation with $\Delta_T\leq\tau$ satisfies $\E\Delta_Y\geq\mathcal L(\tau)$.
\end{theorem}

\begin{proof}
Write $\delta=\pi_A-\pi_B$. The budget gives $\pi_B=\bar\pi-\delta/2$, while capacity and tolerance give $|\delta|\leq\tau_{\mathrm{eff}}$. Equation~\eqref{eq:decomp-thm} yields $\E D=\kappa_2\bar\pi+c\delta$. Its range is centred at $\kappa_2\bar\pi\geq0$ with radius $|c|\tau_{\mathrm{eff}}$, so its distance from zero gives Eq.~\eqref{eq:floor}. Jensen's inequality gives $\E\Delta_Y\geq|\E D|\geq\mathcal L(\tau)$.
\end{proof}

When $c>0$, offsetting the structural term requires a higher seed rate in $B$. When $c<0$, stronger cross-group transmission reverses this direction. At parity, $\mathcal L(0)=\gamma(m-1)\mathcal H\bar\pi$.

The asymmetry $\mathcal H$ determines the parity gap, whereas the slope of the allocation response depends on average within-group connectivity relative to cross-group connectivity:
\begin{equation}
c=\beta+\gamma\left[\frac{m-1}{2}(p_A+p_B)-mq\right].
\label{eq:mixing-slope}
\end{equation}
The same density asymmetry can require opposite allocation corrections. For $c\ne0$, cancelling the structural term requires $|\pi_A-\pi_B|=\kappa_2\bar\pi/|c|$, subject to tolerance and capacity. If $c=0$, reallocation leaves the first moment unchanged.

\begin{proposition}[Fraction exposure]
\label{prop:normalized}
Replace the count exposure in Eq.~\eqref{eq:outcome-model} by $d_i^{-1}\sum_{j\in N_i}T_j$, assigning zero exposure to isolated vertices. Let $\tilde\mu_g$ be the resulting group mean. Under uniform within-group seeding,
\begin{equation}
\E[\tilde\mu_A-\tilde\mu_B]
=\tilde\kappa_1(\pi_A-\pi_B)+\tilde\kappa_2\pi_B,
\label{eq:normalized-decomp}
\end{equation}
where $\tilde\kappa_1=\beta+\gamma(\theta_A-\theta_B)$. Here $\theta_A$ and $\theta_B$ are the expected fractions of $A$-neighbours seen from an $A$ vertex and a $B$ vertex, with zero assigned at isolation. The isolation coefficient is $\tilde\kappa_2=\gamma[\Pr(d_i=0\mid i\in V_B)-\Pr(d_i=0\mid i\in V_A)]$. For balanced groups with $m\geq2$,
\begin{equation}
0\leq\tilde\kappa_2\leq\gamma e^{-\bar d_B},
\qquad \bar d_B=(m-1)p_B+mq.
\label{eq:normalized-bound}
\end{equation}
Thus the structural coefficient vanishes as the sparser group's expected degree diverges.
\end{proposition}

Under fraction exposure, equal rates give equal expected exposure at every non-isolated vertex. Isolation alone carries the structural difference. The absolute gap also includes graph fluctuations, as in Eq.~\eqref{eq:fluctuation}.

\begin{proposition}[Multiple groups]
\label{prop:multigroup}
For a $G$-block SBM with within-block probabilities $p_g$, symmetric cross-block probabilities $q_{gf}$, and graph-independent seeding at fixed group rates,
\begin{equation}
\E\mu_g=\alpha+\beta\pi_g+\gamma\left[(n_g-1)p_g\pi_g+\sum_{f\ne g}n_fq_{gf}\pi_f\right].
\label{eq:multigroup}
\end{equation}
For a pair $(g,h)$, a third block contributes $\gamma n_f(q_{gf}-q_{hf})\pi_f$ to the difference. Uniform cross-block probabilities remove these third-block terms.
\end{proposition}

\section{Decoupled Procedural Allocation}
\label{sec:algorithm}

Model-based DPA variants derive quotas from the structural decomposition. Observed-graph selection uses exact group contributions and quota-preserving exchanges, as illustrated in \Cref{fig:overview}.

\begin{figure*}[t]
\centering
\includegraphics[width=0.90\textwidth]{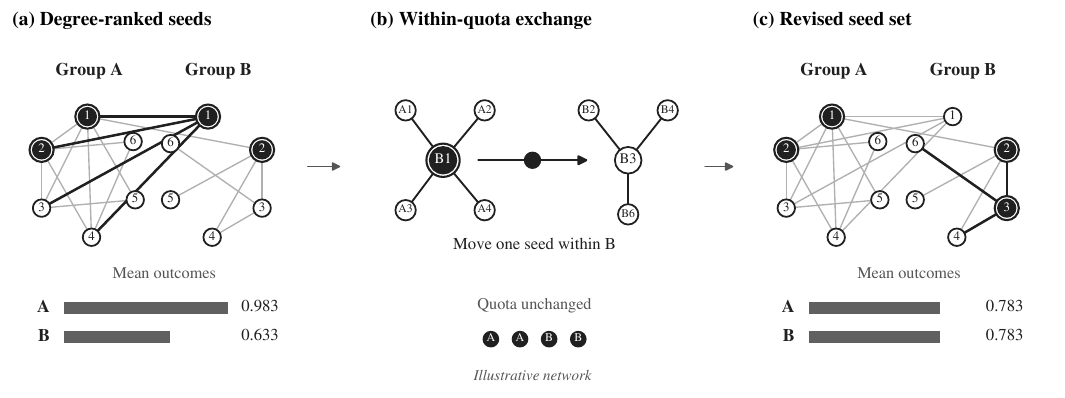}
\Description{An illustrative two-group graph shows degree-ranked seeds and an exchange within the same group quota. The exchanged seed changes the receiving groups of spillovers, and paired bars compare group means before and after the exchange.}
\caption{A within-quota exchange on an illustrative 12-node graph. Four seeds are split equally between the groups. Each seed adds one unit to itself and $0.3$ units to each neighbour, with zero baseline. Replacing $B_1$ with $B_3$ closes the outcome gap while retaining $96.91\%$ of total outcome. Rings mark seeds, and bold edges highlight the exchanged nodes' neighbours.}
\label{fig:overview}
\end{figure*}

\subsection{Closed-Form Quotas}
\label{sec:alg-stage1}

Given coefficient estimates $(\hat\kappa_1,\hat\kappa_2)$, the closed-form quota rule minimises the predicted gap over the continuous rates in Eq.~\eqref{eq:feasible-rates}.
Eliminating $\pi_A$ yields $|a-b\pi_B|$, where $a=\hat\kappa_1k/n_A$ and $b=\hat\kappa_1n/n_A-\hat\kappa_2$.

\begin{proposition}[Closed-form solution]
\label{prop:stage1}
Let the capacity- and tolerance-feasible interval be
\begin{align}
\ell&=\max\left\{0,\frac{k-n_A\tau}{n},\frac{k-n_A}{n_B}\right\},\nonumber\\
u&=\min\left\{1,\frac{k+n_A\tau}{n},\frac{k}{n_B}\right\}.
\label{eq:stage1-interval}
\end{align}
If $b\ne0$, the solution is $\pi_B^\star=\operatorname{proj}_{[\ell,u]}(a/b)$. If $b=0$, every feasible rate is optimal. In both cases $\pi_A^\star=(k-n_B\pi_B^\star)/n_A$.
\end{proposition}

Estimating the structural coefficients from observed group edge densities takes $O(|E|)$ time. The oracle variant uses the generating SBM parameters. Coefficients estimated on the evaluation graph define an adaptive policy whose performance we measure empirically.

\subsection{Within-Group Selection and Guarantees}
\label{sec:alg-stage2}

\textsc{DPA-r} draws the prescribed number of seeds uniformly without replacement within each group. \textsc{DPA-g} chooses the highest-degree nodes. Equation~\eqref{eq:spread} gives each seed's total-outcome contribution as $\beta+\gamma d_i$, so degree ranking maximises $\sigma$ exactly at fixed quotas. Degree-ranked seeds provide the starting point for observed-graph selection.

Convert continuous quotas to integers by $k_A^\star=\lceil n_A\pi_A^\star\rceil$ and $k_B^\star=k-k_A^\star$.

\begin{lemma}[Rounding]
\label{lem:rounding}
The rounded profile $\bm\pi^{\mathrm{int}}$ obeys $\Delta_T\leq\tau+1/n_A+1/n_B$ and
\begin{equation}
|\Phi(\bm\pi^{\mathrm{int}})-\Phi(\bm\pi^\star)|\leq B_n,
\quad B_n=|\kappa_1|\left(\frac1{n_A}+\frac1{n_B}\right)+\frac{|\kappa_2|}{n_B}.
\label{eq:rounding}
\end{equation}
\end{lemma}

\begin{theorem}[Randomised allocation guarantee]
\label{thm:taa}
Let the coefficient estimates be fixed independently of the evaluated graph and satisfy $|\hat\kappa_i-\kappa_i|\leq\eta_i$. For balanced groups, \textsc{DPA-r} satisfies
\begin{equation}
\E\Delta_Y(S^{\mathrm{DPA-r}})
\leq\mathcal L(\tau)+2(\eta_1\tau+\eta_2)+B_n+r_n.
\label{eq:guarantee}
\end{equation}
The returned allocation's access disparity satisfies the rounding bound in \Cref{lem:rounding}. Here $r_n$ is evaluated at the returned allocation profile. With oracle coefficients, the first-moment floor is attained up to rounding, and the expected absolute gap adds at most $r_n$.
\end{theorem}

Degree ranking maximises fixed-quota total outcome because seed contributions add, with fairness evaluated on the realised graph.

\subsection{Within-Quota Selection on an Observed Graph}
\label{sec:alg-dpax}

Seed $v$ adds $a_g(v)=(\beta\mathds{1}[v\in V_g]+\gamma d_g(v))/n_g$ to group $g$'s mean. Its signed contribution $w_v=a_A(v)-a_B(v)$ to the outcome difference is
\begin{equation}
w_v=\beta\left(\frac{\mathds{1}[v\in V_A]}{n_A}-\frac{\mathds{1}[v\in V_B]}{n_B}\right)
+\gamma\left(\frac{d_A(v)}{n_A}-\frac{d_B(v)}{n_B}\right).
\label{eq:node-contribution}
\end{equation}
Thus $D(S)=\sum_{v\in S}w_v$ and $\sigma(S)=n\alpha+\sum_{v\in S}c_v$, where $c_v=\beta+\gamma d_v$. Let $J(\tau)$ contain integer $A$-quotas $j=|S\cap V_A|$ meeting capacity and access constraints. If this set is empty, recompute $J$ at the minimum attainable access gap. Fixed degree rankings define the prefix set $S(j)$, containing the first $j$ nodes in $A$ and the first $k-j$ in $B$. A degree-prefix sweep minimises $|D(S(j))|$ over this family in $O(k)$ time after $O(n\log n+|E|)$ preprocessing. Every prefix maximises total outcome for its quota. \textsc{DPA-e} breaks minimum-gap ties by the smallest $j$.

A quota change replaces the last selected $B$ node $b$ with the next $A$ node $a$, giving
\begin{align}
D(S(j+1))-D(S(j))&=w_a-w_b,\label{eq:prefix-gap-exchange}\\
\sigma(S(j+1))-\sigma(S(j))&=\gamma(d_a-d_b).\label{eq:prefix-spread-exchange}
\end{align}
A within-group exchange preserves $j$ and changes $D$ by incoming minus outgoing contribution. Degree measures aggregate spillover, whereas $w_v$ records its receiving groups, allowing \textsc{DPA} to search beyond degree prefixes.

\begin{proposition}[Within-quota envelope]
\label{prop:envelope}
For quota $j$, let $\ell_j$ sum the $j$ smallest values of $w_v$ in $A$ and the $k-j$ smallest in $B$, and let $u_j$ sum the corresponding largest values. Every seed set at this quota satisfies $D(S)\in[\ell_j,u_j]$. Hence every access-feasible seed set satisfies
\begin{equation}
\Delta_Y(S)\geq\min_{j\in J(\tau)}\operatorname{dist}(0,[\ell_j,u_j]).
\label{eq:envelope-floor}
\end{equation}
A quota cannot reach the outcome band if its interval misses $[-\epsilon,\epsilon]$. When every interval excludes zero, extreme contribution prefixes attain the bound.
\end{proposition}

For a quota whose interval meets the band and a tilt $\nu\geq0$, rank nodes within each group by $c_v-\nu w_v$. Denote the resulting quota-constrained prefix by $S_j(\nu)$.

\begin{lemma}[Tilted rankings and total-outcome bound]
\label{lem:tilt}
For $\nu\geq0$, $S_j(\nu)$ maximises $\sigma(S)-\nu D(S)$ at quota $j$, and its signed gap is non-increasing in $\nu$. Every set at that quota with $D(S)\leq\epsilon$ satisfies
\begin{equation}
\sigma(S)\leq U_j(\nu):=\sigma(S_j(\nu))-\nu[D(S_j(\nu))-\epsilon].
\label{eq:certificate}
\end{equation}
The symmetric construction uses $-w_v$. A basic solution of the fixed-quota linear relaxation of Eq.~\eqref{eq:problem-band} has at most two fractional coordinates, both in the same group.
\end{lemma}

Zero tilt recovers degree selection. Bisection brackets a band crossing, and within-group exchanges repair and improve the candidate. Ranked seed sets change discretely with the tilt, so each returned set is verified from its summed contributions.

\begin{algorithm}[t]
\caption{\textsc{DPA} on an observed graph}
\label{alg:dpax}
\begin{algorithmic}[1]
\Require Graph, groups, budget $k$, tolerances $\tau,\epsilon$, effects $\beta,\gamma$
\State Compute $w_v$, $c_v$, $J(\tau)$ and the envelopes.
\State Initialise candidates with degree and extreme-contribution prefixes.
\For{each evaluated quota $j$}
\If{$[\ell_j,u_j]$ meets $[-\epsilon,\epsilon]$}
\State Bracket the band crossing by tilted rankings and record the bound in Eq.~\eqref{eq:certificate}.
\State Apply within-group repair and improvement swaps, retaining verified in-band candidates.
\Else
\State Record the minimum-gap extreme prefix at quota $j$.
\EndIf
\EndFor
\State \Return the highest-outcome in-band candidate, or the smallest-gap evaluated candidate if none is in band.
\end{algorithmic}
\end{algorithm}

Algorithm~\ref{alg:dpax} searches quotas and returns a verified candidate, prioritising total outcome within the band. Computing contributions costs $O(n+|E|)$, and each tilt uses $O(n)$ partial sorting. Let $U_j$ be a valid bound at every access-feasible quota. For an in-band output $\hat S$ and the optimum $\sigma^\star_{\tau,\epsilon}$ of Eq.~\eqref{eq:problem-band},
\begin{equation}
\sigma(\hat S)\leq\sigma^\star_{\tau,\epsilon}\leq\max_{j\in J(\tau)}U_j.
\label{eq:global-certificate}
\end{equation}
A full bound scan certifies the allocation search as well as the seed selection. The Pokec evaluation includes this scan separately from selection time.

\section{Experiments}
\label{sec:experiments}

In this section, we first describe the setup and baselines, test the structural floor, and compare methods. Then, we examine outcome tolerance and exposure. Finally, we study independent-cascade diffusion, runtime, and allocation on Pokec. Synthetic method comparisons use ten SBM or five BA graphs per suite and report mean outcome gaps and sample standard deviations.

\subsection{Setup}
\label{sec:exp-setup}

\paragraph{Graphs.}
\textbf{SBM-Synth} uses balanced groups, $p_A=0.06$, $q=0.005$, $p_B\in[0.005,0.040]$, and $n\in\{400,1000,2000\}$. Method comparisons use $n=1000$ and $p_A-p_B=0.04$. \textbf{BA-Synth} joins two Barab\'asi--Albert (BA) graphs with attachment parameters $4$ and $1$ and cross-edge probability $0.002$. \textbf{BlogCatalog} and \textbf{Flickr} use preprocessed releases \citep{guo2020networked} of the original networks \citep{tang2009relational}, retaining their two largest labelled communities, while \textbf{Amherst41} \citep{traud2012social} and \textbf{Pokec} \citep{takac2012data} use recorded gender. \Cref{tab:datasets} lists their sizes. Development uses ten SBM graphs, five BA graphs, BlogCatalog, Flickr, and Pokec. Prospective evaluation uses ten further SBM graphs, five further BA graphs, and Amherst41, specified in the pre-evaluation protocol. The bound study uses ten graphs per setting, and other synthetic studies use five.

\begin{table}[t]
\centering
\caption{Graph and group sizes for method comparisons, with grouping attributes.}
\label{tab:datasets}
\footnotesize
\setlength{\tabcolsep}{4pt}
\begin{tabular}{lrrl}
\toprule
Dataset & $n_A$ & $n_B$ & Grouping \\
\midrule
SBM-Synth & $500$ & $500$ & Block \\
BA-Synth & $500$ & $500$ & Barab\'asi--Albert component \\
BlogCatalog & $892$ & $921$ & Community \\
Flickr & $883$ & $888$ & Community \\
Amherst41 & $1\,015$ & $1\,017$ & Profile gender \\
Pokec & $828\,301$ & $804\,335$ & Profile gender \\
\bottomrule
\end{tabular}

\end{table}

\paragraph{Protocol.}
Unless stated otherwise, the baseline outcome is zero, each seed adds one unit to its own outcome and $0.3$ units to each neighbour, $k=\lfloor0.1n\rfloor$, and $\epsilon=\epsilon_0$. Method comparisons use double-precision closed-form totals. Synthetic summaries give means and sample standard deviations across graphs. \textsc{LP-Greedy} uses five algorithm seeds on each fixed graph. All methods share the graph, budget, and outcome model. Real-network outcomes are modelled scores.

\begin{table*}[t]
\centering
\caption{Results with access feasibility assessed at $\tau=0.1$. Gaps show means $\pm$ sample standard deviations under the aggregation in Section~\ref{sec:exp-setup}. Bold and underlining mark the smallest and second-smallest mean gaps among configurations feasible in every run. Daggers indicate an access violation in at least one run. Welfare uses log utility (inequality aversion $\lambda=1$). DPA-r uses one fixed draw on each fixed graph.}
\label{tab:q2}
\footnotesize
\setlength{\tabcolsep}{3pt}
\begin{tabular}{l@{\quad}rrr@{\quad}rrr@{\quad}rrr@{\quad}rrr}
\toprule
& \multicolumn{3}{c}{SBM-Synth} & \multicolumn{3}{c}{BA-Synth} & \multicolumn{3}{c}{BlogCatalog} & \multicolumn{3}{c}{Flickr} \\
Method & $\Delta_T$ & $\Delta_Y$ & $\sigma$ & $\Delta_T$ & $\Delta_Y$ & $\sigma$ & $\Delta_T$ & $\Delta_Y$ & $\sigma$ & $\Delta_T$ & $\Delta_Y$ & $\sigma$ \\
\midrule
Greedy-IM & 0.200$^{\dagger}$ & $2.229{\pm}0.039$ & 1311 & 0.166$^{\dagger}$ & $1.060{\pm}0.025$ & 716 & 0.009 & $0.547$ & 6836 & 0.013 & $0.742$ & 5161 \\
Proportional & 0.000 & $0.732{\pm}0.031$ & 1020 & 0.000 & $0.567{\pm}0.015$ & 643 & 0.000 & $0.402$ & 6835 & 0.001 & $0.636$ & 5154 \\
Maximin & 0.094 & $0.013{\pm}0.008$ & 850 & 0.149$^{\dagger}$ & $0.006{\pm}0.004$ & 468 & 0.031 & $0.006$ & 6802 & 0.086 & $0.007$ & 4887 \\
Welfare & 0.027 & $0.929{\pm}0.058$ & 1065 & 0.019 & $0.610{\pm}0.043$ & 654 & 0.000 & $0.402$ & 6835 & 0.003 & $0.614$ & 5152 \\
\midrule
FIMM & 0.129$^{\dagger}$ & $1.586{\pm}0.062$ & 1211 & 0.091 & $0.793{\pm}0.018$ & 692 & 0.002 & $0.372$ & 6834 & 0.008 & $0.698$ & 5160 \\
LP-Greedy & 0.078 & $0.009{\pm}0.008$ & 865 & 0.130$^{\dagger}$ & $0.010{\pm}0.006$ & 485 & 0.020 & $0.004{\pm}0.004$ & 6813 & 0.064 & $\underline{0.005{\pm}0.000}$ & 4942 \\
IGM-AGM & 0.090 & $\underline{0.005{\pm}0.003}$ & 842 & 0.136$^{\dagger}$ & $0.007{\pm}0.006$ & 481 & 0.002 & $\underline{0.003}$ & 6645 & 0.073 & $0.015$ & 4918 \\
LP-Greedy-q & 0.078 & $0.009{\pm}0.008$ & 865 & 0.100 & $0.019{\pm}0.018$ & 484 & 0.020 & $0.004{\pm}0.004$ & 6813 & 0.064 & $\underline{0.005{\pm}0.000}$ & 4942 \\
IGM-AGM-q & 0.090 & $\underline{0.005{\pm}0.003}$ & 842 & 0.100 & $0.170{\pm}0.016$ & 536 & 0.002 & $\underline{0.003}$ & 6645 & 0.073 & $0.015$ & 4918 \\
\midrule
\textsc{DPA-g} & 0.094 & $0.023{\pm}0.013$ & 851 & 0.079 & $0.303{\pm}0.021$ & 571 & 0.000 & $0.402$ & 6835 & 0.050 & $0.266$ & 5056 \\
\textsc{DPA-r} & 0.094 & $0.021{\pm}0.019$ & 631 & 0.079 & $\underline{0.013{\pm}0.012}$ & 244 & 0.000 & $0.238$ & 2755 & 0.050 & $0.069$ & 1500 \\
\textsc{DPA-e} & 0.096 & $0.010{\pm}0.005$ & 846 & 0.100 & $0.223{\pm}0.018$ & 546 & 0.031 & $0.013$ & 6802 & 0.086 & $0.007$ & 4887 \\
\textsc{DPA} & 0.079 & $\mathbf{(2.8{\pm}1.7){\times}10^{-4}}$ & 863 & 0.100 & $\mathbf{(3.2{\pm}1.1){\times}10^{-4}}$ & 475 & 0.020 & $\mathbf{2.6{\times}10^{-4}}$ & 6812 & 0.066 & $\mathbf{6.4{\times}10^{-5}}$ & 4939 \\
\bottomrule
\end{tabular}

\end{table*}

\subsection{Compared Methods}
\label{sec:exp-baselines}

We compare nine baseline configurations spanning total outcome, access parity, and group-outcome fairness. \textsc{Greedy-IM} maximises total outcome \citep{kempe2003}. \textsc{Proportional} applies population-proportional quotas \citep{farnadi2020unifying}, with high-degree selection within each group. \textsc{Maximin} maximises the lower group mean \citep{tsang2019}. \textsc{Welfare} optimises isoelastic group welfare with inequality-aversion parameter $\lambda\in\{0,0.25,0.5,0.75,1\}$, using logarithmic welfare at $\lambda=1$ \citep{rahmattalabi2021welfare}. Maximin and Welfare enumerate group splits with degree-based selection.

We adapt three recent seed-level baselines to the common linear outcome model. The fair influence maximization algorithm (\textsc{FIMM}) of \citet{rui2025tkde} greedily maximises size-weighted welfare $\sum_g n_g\mu_g^{1/2}$. \textsc{LP-Greedy} samples from an LP over marginal group gains \citep{spaeh2025multiobjective}. Inner-group maximization (IGM) and across-group maximization (AGM) form \textsc{IGM-AGM}, which greedily merges group-specific seed lists to improve the lower group mean \citep{rui2025maximin}. The quota-constrained variants, \textsc{LP-Greedy-q} and \textsc{IGM-AGM-q}, admit selections with an access-feasible completion. \textsc{DPA-g} and \textsc{DPA-r} use quotas estimated from group edge densities and select degree-ranked or uniform random seeds, respectively. \textsc{DPA-e} searches degree prefixes across feasible quotas, isolating the effect of within-quota seed exchanges.

\subsection{Bound Tightness and Allocation Tolerance}
\label{sec:exp-q1}

\Cref{fig:q1} shows that randomised oracle allocation follows the structural prediction across heterogeneity levels at access parity. Its mean relative difference from the floor is $2.1\%$, $2.5\%$, and $0.3\%$ at $n=400$, $1000$, and $2000$.

Relaxing the access constraint reduces the residual. At $n=1000$ and $\mathcal H=0.04$, increasing $\tau$ from zero to $0.05$ lowers the randomised gap from $0.586$ to $0.291$, close to the decrease predicted by the affine decomposition. The continuous floor reaches zero near $\tau=0.096$, after which the absolute gap levels off near $0.028$, consistent with the fluctuation and rounding terms in Eq.~\eqref{eq:guarantee}.

\begin{figure*}[t]
\centering
\includegraphics[width=0.80\textwidth]{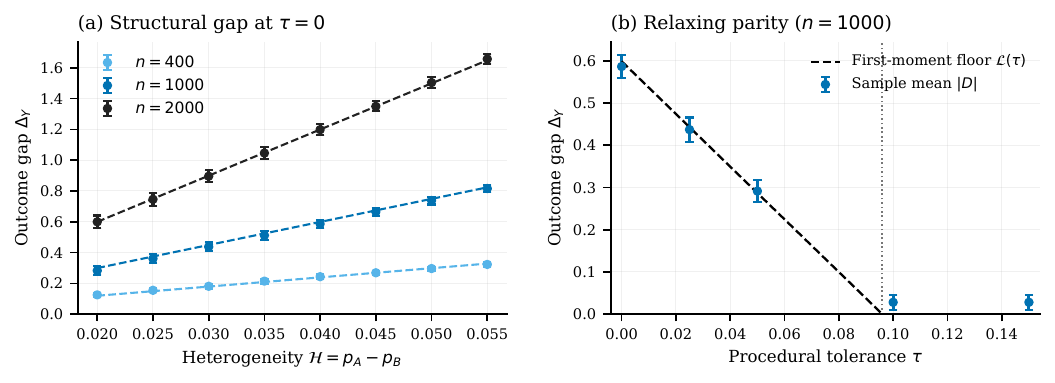}
\Description{Parity outcome gaps increase with density asymmetry, and relaxing the access tolerance lowers the gap towards the first-moment floor.}
\caption{Randomised oracle allocation follows the structural prediction across heterogeneity at parity (left) and access tolerance (right). Dashed lines show the continuous first-moment floor on the expected gap, the dotted line marks its zero threshold, and bars show sample standard deviations over ten graph and allocation realisations.}
\label{fig:q1}
\end{figure*}

\subsection{Fairness Trade-offs on Synthetic and Web Graphs}
\label{sec:exp-q2}

As \Cref{tab:q2} shows, DPA has the smallest mean outcome gap among configurations feasible in every run, at most $3.2\times10^{-4}$ on the four small development settings at $\tau=0.1$. \Cref{fig:q2} shows how the allocation rule changes the attainable trade-off as access tolerance varies.

\paragraph{Quota choice and seed composition.}
Within-quota selection reduces disparity beyond degree prefixes. On BA-Synth, \textsc{DPA-e} leaves a gap of $0.223$ at the common tolerance. \textsc{LP-Greedy-q} reduces it to $0.019$, and \textsc{DPA} reaches $3.2\times10^{-4}$ with $1.7\%$ less total outcome than \textsc{LP-Greedy-q}. On BlogCatalog, quotas based on average mixing remain near access parity. Selecting seeds by their empirical group contributions reduces the gap within the same access constraint.

\paragraph{The cost of parity.}
At the closest access-parity quota, \textsc{DPA} meets the outcome tolerance on BlogCatalog and Flickr while retaining more than $90\%$ of their degree allocations' total outcome. On BA-Synth, DPA reaches the band, retaining a mean $66.4\%$ of the degree allocation's total outcome. On SBM-Synth, every parity-quota envelope excludes zero, and the extreme prefixes attain the minimum gap at that quota, averaging $0.130$, below the first-moment floor $\mathcal L(0)=0.599$ for graph-independent seeding.

\paragraph{Prospective evaluation and exact references.}
\Cref{tab:dpax-prospective} shows the smallest access-feasible mean gap for DPA in all three prospective settings. On Amherst41, DPA nearly equalises modelled gender outcomes at the closest access-parity quota and retains $99.8\%$ of the degree-based allocation's total outcome. Across both suites, all 330 evaluations on small graphs meet their reference band, with total-outcome shortfalls below $0.41\%$ of the band-constrained integer optimum. Here the reference band is $\Delta_Y(S)\leq\max(\epsilon,\Delta_{\mathrm{ref}})+10^{-9}$, where $\Delta_{\mathrm{ref}}$ is the incumbent gap from disparity minimisation.

\begin{table}[t]
\centering
\caption{Prospective outcome gaps $\Delta_Y$ at $\tau=0.1$ for all thirteen configurations. Aggregation and rank markings follow \Cref{tab:q2}. Daggers flag access violations.}
\label{tab:dpax-prospective}
\footnotesize
\setlength{\tabcolsep}{2pt}
\begin{tabular}{lrrr}
\toprule
Method & SBM (new) & BA (new) & Amherst41 \\
\midrule
Greedy-IM & $2.255{\pm}0.035^{\dagger}$ & $1.032{\pm}0.039^{\dagger}$ & $0.258$ \\
Proportional & $0.737{\pm}0.025$ & $0.555{\pm}0.007$ & $0.218$ \\
Maximin & $0.013{\pm}0.008$ & $0.006{\pm}0.006^{\dagger}$ & $0.137$ \\
Welfare & $0.939{\pm}0.065$ & $0.611{\pm}0.030$ & $0.220$ \\
\midrule
FIMM & $1.576{\pm}0.060^{\dagger}$ & $0.776{\pm}0.018$ & $0.210$ \\
LP-Greedy & $0.010{\pm}0.008$ & $0.004{\pm}0.003^{\dagger}$ & $0.001{\pm}0.000$ \\
IGM-AGM & $0.009{\pm}0.006$ & $0.005{\pm}0.006^{\dagger}$ & $0.003$ \\
LP-Greedy-q & $0.010{\pm}0.008$ & $\underline{0.004{\pm}0.003}$ & $0.001{\pm}0.000$ \\
IGM-AGM-q & $0.009{\pm}0.006$ & $0.153{\pm}0.030$ & $0.003$ \\
\midrule
\textsc{DPA-g} & $0.026{\pm}0.020$ & $0.285{\pm}0.021$ & $0.015$ \\
\textsc{DPA-r} & $0.031{\pm}0.014$ & $0.010{\pm}0.008$ & $0.022$ \\
\textsc{DPA-e} & $\underline{0.007{\pm}0.003}$ & $0.204{\pm}0.020$ & $\underline{4.8{\times}10^{-4}}$ \\
\textsc{DPA} & $\mathbf{(2.2{\pm}1.5){\times}10^{-4}}$ & $\mathbf{(2.0{\pm}0.0){\times}10^{-4}}$ & $\mathbf{1.8{\times}10^{-4}}$ \\
\bottomrule
\end{tabular}

\end{table}

\begin{figure*}[t]
\centering
\includegraphics[width=0.88\textwidth]{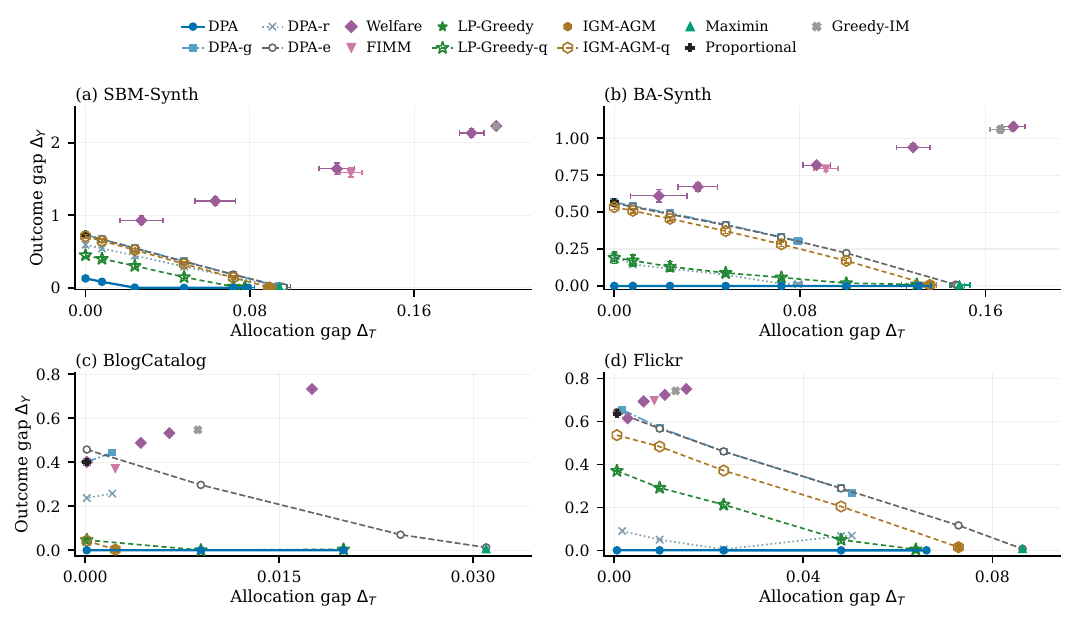}
\Description{Four panels compare the access and outcome gaps of DPA, the degree-prefix variant, and the baselines on the development graphs.}
\caption{Allocation--outcome trade-offs for all thirteen configurations. DPA variants vary $\tau$, and Welfare varies $\lambda$. Bars give one standard deviation across ten SBM or five BA graphs. Fixed-graph LP-Greedy variants average five algorithm seeds. Coincident curves overlap.}
\label{fig:q2}
\end{figure*}

\subsection{Outcome Tolerance and Exposure}
\label{sec:exp-q3}

\paragraph{Outcome tolerance.}
All 102 evaluations of $\epsilon\in\{10^{-4},\epsilon_0,2\epsilon_0\}$ across both suites satisfy their requested outcome bands. On the seven small-graph settings, tightening the tolerance to $10^{-4}$ reduces mean total outcome by at most $0.303\%$, while doubling $\epsilon_0$ increases it by at most $0.051\%$. On Pokec, $10^{-4}$ relaxes the default band and increases total outcome by $0.0009\%$.

\paragraph{Count versus fraction exposure.}
\label{sec:exp-q4}
For the same randomised seed sets, the count-based gap rises with $\mathcal H$, while the fraction-based gap remains around $10^{-3}$. The two community-labelled Web graphs show the same ordering under integer-rounded parity quotas. This ordering matches the isolation coefficient in \Cref{prop:normalized}.

\subsection{Access Constraints under Multi-Hop Diffusion}
\label{sec:exp-ic}

Under fixed within-group degree rankings, the access-constrained activation gaps meet the persistence criterion on SBM-Synth, BA-Synth, and Flickr at $p_{\mathrm{IC}}=0.05$, as \Cref{tab:ic-gate} shows. No setting meets it at $p_{\mathrm{IC}}=0.01$, so propagation strength changes the allocation--outcome trade-off. Here $p_{\mathrm{IC}}$ is the per-edge transmission probability, and quotas vary under fixed rankings to isolate access constraints under independent-cascade diffusion. Let $F_{\mathrm{bind}}$ and $F_{\mathrm{slack}}$ denote the minimum estimated activation gaps over the quota grid under $\Delta_T\leq0.01$ and $\Delta_T\leq0.5$, respectively. Synthetic entries average per-graph minima. The noise unit $s$ is the larger standard deviation of the two minima across graphs or cascade batches. Each seed set uses 2000 cascades, with a prespecified extension for unresolved cases.

\begin{table}[t]
\centering
\caption{IC gaps under fixed within-group degree rankings, with access tolerances $\tau=0.01$ for $F_{\mathrm{bind}}$ and $0.5$ for $F_{\mathrm{slack}}$. Persistence requires $F_{\mathrm{bind}}\geq F_{\mathrm{slack}}+s$, $F_{\mathrm{bind}}\geq1.5F_{\mathrm{slack}}$, and $F_{\mathrm{bind}}\geq0.03$. BlogCatalog uses the 4000-cascade extension.}
\label{tab:ic-gate}
\footnotesize
\begin{tabular}{l|cccl|c}
\toprule
& \multicolumn{4}{c|}{$p_{\mathrm{IC}}{=}0.05$} & $p_{\mathrm{IC}}{=}0.01$ \\
Graph & $F_{\mathrm{bind}}$ & $F_{\mathrm{slack}}$ & $s$ & verdict & $F_{\mathrm{bind}}$ \\
\midrule
SBM-Synth & 0.3858 & 0.2386 & 0.0081 & passes & 0.0209 \\
BA-Synth & 0.0706 & 0.0069 & 0.0029 & passes & 0.0060 \\
BlogCatalog & 0.0123 & 0.0072 & 0.0014 & below threshold & 0.0059 \\
Flickr & 0.1384 & 0.0918 & 0.0006 & passes & 0.0201 \\
\bottomrule
\end{tabular}

\end{table}

\subsection{Computational Scale and the Pokec Case Study}
\label{sec:exp-scale}

\paragraph{Runtime.}
The degree-prefix sweep averages $0.212$ s on one million sparse-SBM nodes by reusing contributions across quotas. \Cref{tab:dpax-runtime} compares all thirteen configurations. \textsc{DPA} takes $20.9$ s on Pokec, against $79.6$ s for \textsc{LP-Greedy}.

\begin{table}[t]
\centering
\caption{Selection time (s) includes preparation and excludes loading and shared degree computation. Small graphs use a $10\%$ budget and $\tau=0.1$, and Pokec $1\%$ and $\tau=0.01$. DPA-r$^{\ddagger}$ timings average five calls per graph after one warm-up. Other methods are timed once per algorithm run.}
\label{tab:dpax-runtime}
\footnotesize
\setlength{\tabcolsep}{2pt}
\begin{tabular}{lrrrrr}
\toprule
Method & SBM & BA & BlogCatalog & Flickr & Pokec \\
\midrule
Greedy-IM & 0.0001 & 0.00004 & 0.0001 & 0.0001 & 0.015 \\
Proportional & 0.0002 & 0.0001 & 0.0002 & 0.0002 & 0.032 \\
Maximin & 0.006 & 0.003 & 0.007 & 0.007 & 293.5 \\
Welfare & 0.005 & 0.003 & 0.009 & 0.008 & 390.8 \\
\midrule
FIMM & 0.002 & 0.001 & 0.003 & 0.002 & 208.8 \\
LP-Greedy & 0.035 & 0.024 & 0.086 & 0.048 & 79.6 \\
IGM-AGM & 0.002 & 0.001 & 0.002 & 0.002 & 1.57 \\
LP-Greedy-q & 0.037 & 0.031 & 0.082 & 0.053 & 158.2 \\
IGM-AGM-q & 0.002 & 0.001 & 0.002 & 0.002 & 91.9 \\
\midrule
\textsc{DPA-g} & 0.0004 & 0.0003 & 0.0009 & 0.0005 & 0.381 \\
\textsc{DPA-r}$^{\ddagger}$ & 0.0002 & 0.0001 & 0.0004 & 0.0003 & 0.288 \\
\textsc{DPA-e} & 0.0004 & 0.0003 & 0.0009 & 0.0005 & 0.637 \\
\textsc{DPA} & 0.190 & 0.371 & 2.42 & 1.78 & 20.9 \\
\bottomrule
\end{tabular}

\end{table}

\paragraph{Gender-based allocation on Pokec.}
\label{sec:case-study}
We symmetrise the friendship graph, remove users without recorded gender, and retain its largest connected component, with $1{,}632{,}636$ users and $22{,}301{,}602$ undirected edges. With a $1\%$ budget and $\gamma=0.3$, the closest-parity degree allocation has an outcome gap of $0.100$. \Cref{tab:pokec} shows that DPA at the same quota reduces it to $2.4\times10^{-7}$ while retaining $99.14\%$ of total outcome. LP-Greedy-q ranks second at this quota, retaining $99.06\%$ with a mean gap about two orders of magnitude larger than DPA's. At $\tau=0.01$, DPA retains $99.67\%$ and gives the smallest numerical gap among the compared methods. Scanning upper bounds over all 8161 feasible quotas certifies a total-outcome shortfall below $0.198$ outcome units.

The recorded DPA-r draw retains only $12.04\%$ of the degree allocation's total outcome at this same quota. With budget and group counts fixed, differences in total outcome follow entirely from selected-node degrees. On this graph, seed composition therefore controls aggregate benefit within the access constraint.

A small positive outcome tolerance permits substantially higher total benefit than exact outcome parity on this graph. At the experimental coefficients, coprime group sizes force equal group means onto a $0.1$ lattice, bounding total outcome by $91.4\%$ of the parity degree allocation.

\begin{table}[t]
\centering
\caption{Pokec with a $1\%$ budget: seven methods at the closest parity quota and twelve rows at $\tau=0.01$. Allocation gaps are in percentage points. Proportional supplies the retained-outcome reference in both blocks. Welfare, LP-Greedy uncertainty, and rank markings follow \Cref{tab:q2}. Daggers flag unrounded access violations.}
\label{tab:pokec}
\footnotesize
\begin{tabular}{lrrr}
\toprule
Method & $\Delta_T$ (pp) & $\Delta_Y$ & Retained $\sigma$ \\
\midrule
\multicolumn{4}{l}{Closest integer parity quota ($\tau=0$)} \\
Proportional & 0.00 & $0.100$ & 100.00\% \\
LP-Greedy-q & 0.00 & $\underline{(3.4{\pm}1.9){\times}10^{-5}}$ & 99.06\% \\
IGM-AGM-q & 0.00 & $0.045$ & 96.94\% \\
\textsc{DPA-g} & 0.00 & $0.100$ & 100.00\% \\
\textsc{DPA-r} & 0.00 & $0.006$ & 12.04\% \\
\textsc{DPA-e} & 0.00 & $0.099$ & 100.00\% \\
\textsc{DPA} & 0.00 & $\mathbf{2.4{\times}10^{-7}}$ & 99.14\% \\
\midrule
\multicolumn{4}{l}{Access tolerance $\tau=0.01$} \\
Greedy-IM & 0.13 & $0.083$ & 100.04\% \\
Welfare & 0.24 & $0.067$ & 100.01\% \\
FIMM & 0.16 & $0.069$ & 100.03\% \\
Maximin & 0.74 & $1.9{\times}10^{-5}$ & 99.00\% \\
LP-Greedy & 0.35 & $(2.6{\pm}4.4){\times}10^{-5}$ & 99.67\% \\
IGM-AGM & 0.41 & $\underline{7.5{\times}10^{-6}}$ & 96.70\% \\
LP-Greedy-q & 0.35 & $(2.6{\pm}4.4){\times}10^{-5}$ & 99.67\% \\
IGM-AGM-q & 0.41 & $\underline{7.5{\times}10^{-6}}$ & 96.70\% \\
\textsc{DPA-g} & 1.00$^{\dagger}$ & $0.040$ & 97.85\% \\
\textsc{DPA-r} & 1.00$^{\dagger}$ & $0.003$ & 12.24\% \\
\textsc{DPA-e} & 0.74 & $1.3{\times}10^{-5}$ & 99.01\% \\
\textsc{DPA} & 0.35 & $\mathbf{6.8{\times}10^{-8}}$ & 99.67\% \\
\bottomrule
\end{tabular}

\end{table}

\section{Conclusion}
\label{sec:discussion}

Our exact decomposition and sharp first-moment floor relate access to propagated benefits under graph-independent seeding. DPA combines access-feasible quotas, within-quota exchanges, contribution envelopes, and total-outcome certificates. It yields the smallest mean gap at the common access tolerances on both suites with near-optimal band-constrained total outcome. The certificates bound remaining total-outcome improvement within each quota and, after a full quota scan, over all access-feasible allocations. On Pokec, DPA nearly equalises modelled gender outcomes while preserving aggregate benefit. Quotas, seed composition, and propagation shape both access and spillover benefits. Future work can learn intervention-specific spillover effects and extend the sharp floor to multiple groups.

\label{main:end}
\clearpage
\bibliographystyle{ACM-Reference-Format}
\bibliography{main}

\end{document}